\documentclass[letter,UKenglish]{lipics-v2018}

\usepackage{amsmath, amssymb, amsthm, amsfonts}
\usepackage{bbm}
\usepackage{bm}
\usepackage{cite}
 \usepackage{thm-restate} 
 
\usepackage[framemethod=tikz]{mdframed}
\usepackage[bottom]{footmisc}
\usepackage{enumitem}
\setitemize{noitemsep,topsep=3pt,parsep=3pt,partopsep=3pt}

\usepackage{tikz}
\usetikzlibrary{arrows,shapes,trees}

\newcommand{\f}{\frac}
\newcommand{\cd}{\cdot}

\newcommand{\sr}{\sqrt}

\newcommand{\lds}{\ldots}

\newcommand{\bs}{\backslash}

\newcommand{\be}{\begin{enumerate}}
\newcommand{\ee}{\end{enumerate}}
\newcommand{\im}{\item}
\newcommand{\bi}{\begin{itemize}}
\newcommand{\ei}{\end{itemize}}

\newcommand{\Sum}{\displaystyle\sum\limits}

\newcommand{\logn}{\log n}

\newcommand{\R}{\mathbb R}

\newcommand{\eps}{\epsilon}
\newcommand{\e}{\epsilon}
\newcommand{\de}{\delta}

\newcommand{\bt}{\beta}

\newcommand{\Om}{\Omega}
\newcommand{\el}{\ell}
\newcommand{\Th}{\Theta}

\newcommand{\Ra}{\Rightarrow}

\newcommand{\lf}{\lfloor}
\newcommand{\rf}{\rfloor}
\newcommand{\lc}{\lceil}
\newcommand{\rc}{\rceil}

\newcommand{\E}{\mathbb E}

\newcommand{\poly}{\textup{poly}}
\newcommand{\polylog}{\textup{polylog}}

\newcommand{\lp}{\left(}
\newcommand{\rp}{\right)}

\newcommand{\lmt}{\left[\begin{matrix}} 
\newcommand{\rmt}{\end{matrix}\right]}

\newcommand{\tO}{\tilde{O}}
\newcommand{\tOm}{\tilde{\Omega}}
\newcommand{\sqn}{\sqrt{n}}
\newcommand{\tODn}{\tO(\sqn + D)}
\newcommand{\mc}{\mathcal}

\nolinenumbers

\newtheorem{defn}{Definition}

\title{Faster Distributed Shortest Path Approximations via Shortcuts}

\author{Bernhard Haeupler}{Carnegie Mellon University}{cs.cmu.edu/~haeupler}{}{Supported in part by NSF grants CCF-1527110, CCF-1618280 and NSF CAREER award CCF-1750808.}

\author{Jason Li}{Carnegie Mellon University}{cs.cmu.edu/~jmli}{}{}

\authorrunning{B. Haeupler and J. Li}
\titlerunning{Faster Distributed Shortest Path Approximations via Shortcuts}

\Copyright{Bernhard Haeupler and Jason Li}

\subjclass{G.2.2 Graph Theory - Graph Algorithms}
\keywords{Distributed Graph Algorithms, Shortest Path, Shortcuts}

\EventEditors{Ulrich Schmid and Josef Widder}
\EventNoEds{2}
\EventLongTitle{32nd International Symposium on Distributed Computing (DISC 2018)}

\EventShortTitle{DISC 2018}
\EventAcronym{DISC}
\EventYear{2018}

\EventDate{October 15--19, 2018}
\EventLocation{New Orleans, USA}
\EventLogo{}
\SeriesVolume{121}
\ArticleNo{35}

\begin{document}

\date{}
\maketitle

\begin{abstract}
A long series of recent results and breakthroughs have led to faster and better distributed approximation algorithms for single source shortest paths (SSSP) and related problems in the CONGEST model. The runtime of all these algorithms, however, is $\tilde{\Omega}(\sqrt{n})$, regardless of the network topology\footnote{We use $\tilde{\ }$-notation to hide polylogarithmic factors in $n$, e.g., $\tilde{O}(f(n)) = O(f(n)\log^{O(1)} n)$.}, even on nice networks with a (poly)logarithmic network diameter $D$. While this is known to be necessary for some pathological networks, most topologies of interest are arguably not of this type.

\smallskip 

We give the first distributed approximation algorithms for shortest paths problems that adjust to the topology they are run on, thus achieving significantly faster running times on many topologies of interest. The running time of our algorithms depends on and is close to $Q$, where $Q$ is the quality of the best shortcut that \emph{exists} for the given topology. While $Q = \tilde{\Theta}(\sqrt{n} + D)$ for pathological worst-case topologies, many topologies of interest\footnote{For example, \cite{ghaffari2016distributed} and \cite{haeupler2016near} show that large classes of interesting network topologies, including planar networks, bounded genus topologies, and networks with polylogarithmic treewidth have shortcuts of quality $Q=\tilde{O}(D)$. A similar statment is likely to hold for any minor closed graph family\cite{minorspersonalcommunication}.} have $Q = \tilde{\Theta}(D)$, which results in near \emph{instance optimal} running times for our algorithm, given the trivial $\Omega(D)$ lower bound. 


\smallskip 

The problems we consider are as follows: 
\begin{itemize}
        \item an approximate shortest path tree and SSSP distances, 
        \item a polylogarithmic size distance label for every node such that from the labels of any two nodes alone one can determine their distance (approximately), and
        \item an (approximately) optimal flow for the transshipment problem. 
\end{itemize}

\smallskip 

Our algorithms have a tunable tradeoff between running time and approximation ratio. Our fastest algorithms have an arbitrarily good polynomial approximation guarantee and an essentially optimal $\tilde{O}(Q)$ running time. On the other end of the spectrum, we achieve polylogarithmic approximations in {$\tilde{O}(Q \cdot n^{\eps})$} rounds for any $\eps > 0$. It seems likely that eventually, our non-trivial approximation algorithms for the SSSP tree and transshipment problem can be bootstrapped to give fast $Q \cdot 2^{O(\sqrt{\log n\log \log n})}$ round $(1+\eps)$-approximation algorithms using a recent result by Becker et al.

%
%
        %
%
        %
%

\end{abstract}

\setcounter{page}{0}
\thispagestyle{empty}

\newpage

\section{Introduction}

This paper gives new distributed approximation algorithms for computing single source shortest path (SSSP) distances and various generalizations, such as computing a SSSP tree, distance labels, and a min-cost uncapacitated flow. 

\smallskip

In the last few years, CONGEST algorithms for shortest path problems have seen a tremendous amount of interest and progress~\cite{nanongkai2014distributed,henzinger16almost,becker2016near}. The main difference of the algorithms developed here, compared to those works, is that our algorithms achieve significantly faster running times for non-pathological network topologies by building on the recently developed~\cite{ghaffari2016distributed,haeupler2016low} low-congestion shortcut framework; for a detailed overview, see Appendix~\ref{sec:shortcutframework} of the full version on arXiv.

\smallskip

The low-congestion shortcut framework leads to faster algorithms for optimization problems with simple parallel divide and conquer style algorithms, such as the minimum spanning tree problem. However it initially seemed less applicable to shortest path problems, particularly because all previous approaches for CONGEST algorithms for these problems led to $\Omega(\sqrt{n})$ running times, for reasons that are independent of issues where shortcuts can help. Indeed, our approach for achieving non-trivial approximation ratios for shortest path problems deviates notably from these approaches, and uses different tools to obtain non-trivial approximation guarantees. 

\smallskip


This paper is organized as follows: We briefly summarizes the key technical concepts of the shortcut framework in Section~\ref{sec:shortcutsummary}; a more detailed treatment of the framework is given in Appendix~\ref{sec:shortcutframework} in the full version. In Section~\ref{sec:SSSPProblems}, we define the different problems we treat in this paper, and explain the difficulties in beating the $\tOm(\sqn + D)$ barrier for approximating shortest path distances. We state our results in Section~\ref{sec:results}, compare it to related works in Section~\ref{sec:related}, and devote the remaining paper to describing our algorithms and proving them correct.

\subsection{The Low-Congestion Shortcut Framework: A Brief Summary}\label{sec:shortcutsummary}

This section provides the key technical definitions and facts about the low-congestion shortcut framework. However, it does not attempt to explain the reasons, generality or importance behind the definitions given here. Appendix~\ref{sec:shortcutframework} of the full version gives a more detailed treatment, and we highly recommend to readers  not familiar with the low-congestion shortcut framework to read Appendix~\ref{sec:shortcutframework} first. 

\smallskip

The shortcut framework is built around a simple and basic communication problem, given in the next two definitions:

\begin{defn}[Valid Partitioning and Parts]
        For a graph $G=(V,E)$, we say that a collection of parts $S_1, S_2, \ldots \subset V$ is a \textbf{valid partition} if the parts are vertex disjoint and each induces a connected graph. 
\end{defn}

\begin{defn}[The Part-wise Communication Problem]\label{def:partwisecomprob}
        Let $G$ be a network with a valid partitioning $S_1,S_2,\lds$ and a value $x_v$ for every node $v \in V$. Suppose $\oplus$ is an associative and commutative function. The \textbf{partwise communication problem} asks for every $S_i$ and every $u \in S_i$ to compute the value $\displaystyle\bigoplus_{v\in S_i} x_v$.
\end{defn}

We remark that for convenience, the parts of a valid partition do not necessarily need to contain every vertex in $V$. Alternatively, it can be convenient to think of each node in $V \setminus \bigcup_i S_i$ as forming its own single-vertex part, thus making any valid partitioning a partitioning in the usual sense. 

\smallskip

The key findings of the shortcut framework can now be summarized as follows:

\begin{mdframed}[roundcorner=4pt, backgroundcolor=white]
The shortcut framework allows us to characterize how hard it is to solve the part-wise communication problem described in Definition~\ref{def:partwisecomprob} in the CONGEST model for any given topology $G$. For any network with topology $G$ this is captured by the  quantity $Q_G$. In the worst-case, the value of $Q_G$ is $\tilde{\Theta}(\sqn+D)$ for a network with $n$ nodes and diameter $D$, such as the pathological network that shows a $\tilde \Omega(\sqrt n+D)$ lower bound for MST and related problems~\cite{sarma2012distributed}. In many other networks of interest, including planar networks, networks which embed into a surface with bounded or polylogarithmic genus, networks with bounded or polylogarithmic tree-width or networks with small separators,  the hardness $Q_G$ is much lower and in fact only $\tO(D)$. Most importantly, whatever the hardness $Q_G$ of a given topology is, there is a simple distributed algorithm which solves the part-wise communication problem in $\tO(Q_G)$ rounds for any valid partitioning in $G$. Thus, $\tO(Q_G)$ round shortcut-based algorithms necessarily have a worst-case running time of $\tO(\sqn + D)$ when expressed in terms of $n$ and $D$; however, they are essentially running as fast as the given topology (and to some extent even the given input) allows it, which in many cases of interest is significantly faster, e.g., $\tO(D)$ rounds.
\end{mdframed}



\smallskip

\subsection{CONGEST model and Shortest Path Problems}\label{sec:SSSPProblems}

\subsubsection{CONGEST Model}

We consider the classical CONGEST model of distributed computing where a network is given by a connected graph $G=(V,E)$ with $n$ nodes and (hop-)diameter $D$. Communication proceeds in synchronous rounds. In each round, each node can send a different $O(\log n)$ bit message to each of its neighbors. Local computations are free and require no time. Nodes have no initial knowledge of the topology $G$, except that we assume that they know $n$ and $D$ up to constants (because these parameters can be computed in $O(D)$ time, which is negligible in our context). All of our algorithms are randomized and succeed with high probability\footnote{Throughout this work, ``with high probability'' or w.h.p. means with probability at least $1 - n^{-C}$ for any desired constant $C$.}. In particular, we assume that each node has access to a private string of randomness, which it can also use to create an $O(\log n)$ bit ID that is unique w.h.p.

In all problems considered here, we assume that every edge $e$ of the network $G$ has a length or cost $w(e)$ associated with it. We assume that all lengths lie in the range $[1,n^{C}]$ for some constant $C$, and are initially only known to nodes adjacent to an edge. Interestingly, our algorithms  also easily handle edges of length zero, but for sake of simplicity, we do not consider such edges in this paper. Any such length or cost function $w$ produces a weighted graph which we call $G(w)$, and induces a distance between any two nodes $u,v \in V$, which we denote with $d_{G(w)}(u,v)$, or simply $d(u,v)$ when the weighted graph $G(w)$ is clear. We denote the weighted diameter of a network with $L = \max_{u,v} d_G(u,v)$.

\subsubsection{Shortest Path Problems}

The most important and most basic problem we are studying in this paper is the single source shortest path problem:

\begin{defn}
        The $\mathbf{\alpha}$\textbf{-approximate SSSP distance problem} assumes as input a weighted graph and a designated source node $s \in V$, and asks for every node $v \in V$ to compute an approximate distance $d_v$ which satisfies $d(s,v) \leq d_v \leq \alpha\cdot d(s,v)$.
\end{defn}
 




We furthermore consider the following generalizations of the SSSP distance problem:

\begin{defn}
        The $\mathbf{\alpha}$\textbf{-approximate SSSP tree problem} assumes that a weighted graph with a designated source node $s \in V$ is given and asks to compute a subtree $T \subseteq G$ such that for every node $v \in V$ distance $d_T(s,v) \leq \alpha \cdot d(s,v)$. Each node should know which of its adjacent edges belong to $T$. 
\end{defn}

\begin{defn}[Approximate distance labeling scheme]
        An \textbf{$(l(n),\alpha)$-approximate distance labeling scheme} is a function that labels the vertices of an input graph with distinct labels up to $l(n)$ bits, such that there exists a polynomial time algorithm that, given the labels of vertices $x$ and $y$, provides an estimate $\tilde d(x,y)$ for the distance between these vertices such that
$$ \tilde d(x,y)\le d(x,y)\le \alpha\cd\tilde d(x,y) .$$
\end{defn}

\begin{defn}[Transshipment Problem]
        The \textbf{transshipment problem} is the problem of uncapacitated min-cost flow. In it every node in a weighted graph $G$ has some real demand $d_v$ such that $\sum_v d_v = 0$. The cost of routing $x$ amount of flow over an edge $e$ of weight $w(e)$ is $x w(e)$. The problem is to compute a flow satisfying all demands of approximate minimum cost. Each node should know the flow an all edges incident to it. 
\end{defn}

\subsection{Our Results}\label{sec:results}

%
%

\subsubsection{SSSP}

Our first result is on computing an approximate, single source shortest path tree in a distributed setting. Note that due to communication limits in the CONGEST model, it is infeasible for each vertex to know the entire shortest path tree. However, it is sufficient that each vertex computes the \textit{local} structure of the tree, which is made specific below.

\begin{restatable}{theorem}{SSSP}\label{SSSP}
Let $G$ be a network graph with edge weights in $ [1,\poly(n)]$, with a specified source vertex, and let $\bt:= (\logn)^{-\Om(1)}$. There is a distributed algorithm that, w.h.p., runs for  $\tilde O(\f1\bt Q_G)$ rounds and outputs a spanning tree that approximates distances to the source to factor $O(L^{O(\log\logn)/\log(1/\bt)})$.\footnote{Recall that $L = \max_{u,v} d_G(u,v)$.} By output, we mean that at the end of the algorithm, every vertex knows its set of incident edges in the spanning tree.
\end{restatable}

By setting $\bt:=n^{-\e}$, $\bt:=2^{-\Th(\sr{\logn})}$, and $\bt:=\log^{-\Th(1/\e)}n$ for constant $\e$, respectively, we obtain the following three corollaries:

\begin{corollary}
Let $G$ be a network graph with edge weights in $ [1,\poly(n)]$, with a specified source vertex. For any constant $\e>0$, there is a distributed algorithm that, w.h.p., runs for  $\tilde O(Q_Gn^\e)$ rounds and outputs a spanning tree that approximates distances to the source to factor $\polylog(n)$.
\end{corollary}

\begin{corollary}
Let $G$ be a network graph with edge weights in $ [1,\poly(n)]$, with a specified source vertex. There is a distributed algorithm that, w.h.p., runs for   $\tilde O(Q_G2^{O(\sqrt{\logn})})$ rounds and outputs a spanning tree that approximates distances to the source to factor $2^{O(\sqrt{\logn})}$. 
\end{corollary}

\begin{corollary}
Let $G$ be a network graph with edge weights in $ [1,\poly(n)]$, with a specified source vertex. For any constant $\e>0$, there is a distributed algorithm that, w.h.p., runs for  $\tilde O(Q_G)$ rounds and outputs a spanning tree that approximates distances to the source to factor $O(L^{\e})$. \end{corollary}

\subsubsection{Distance labeling schemes}

For distance labeling schemes, we have the following result.

\begin{restatable}{theorem}{DistLabel}\label{DistLabel}
Let $G$ be a network graph with edge weights in $[1,\poly(n)]$. There exists a $(\polylog(n),\allowbreak n^{{O(\log\logn)}/{\log(1/\bt)}})$ approximate distance labeling scheme that runs in  $\tilde{O}(\f1\bt Q_G)$ rounds.
\end{restatable}
Setting $\bt:=n^\e$ gives the following corollary:

\begin{corollary}
Let $G$ be a network graph with edge weights in $[1,\poly(n)]$, There exists a $(\polylog(n),\allowbreak\polylog(n))$ approximate distance labeling scheme that runs in  $\tilde{O}(Q_Gn^\e)$ rounds.
\end{corollary}

\subsubsection{Transshipment problem}

We also provide a distributed algorithm to compute an approximate flow for the transshipment problem.

\begin{restatable}{theorem}{TS} \label{TS}
Let $G$ be a network graph with edge weights in $ [1,\poly(n)]$ and demands that sum to zero, and let $\bt:= (\logn)^{-\Om(1)}$. There is an algorithm that, w.h.p., runs for  $\tilde{O}(\f1\bt Q_G)$ rounds and computes a $\tilde O(\f1\bt n^{O(\log\logn)/\log(1/\bt)})$-approximate flow.
\end{restatable}

\subsection{Related Work}\label{sec:related}

The complexity theoretic issues in the design of distributed graph algorithms for the CONGEST model have received much attention in the last decade, and extensive progress has been made for many problems: Minimum-Spanning Tree \cite{kutten1995fast}, Minimum Cut \cite{nanongkai2014almost}, Diameter \cite{lenzen2015fast}, Shortest Path \cite{becker2016near}, and so on. Most of those problems have $\tilde \Th(\sr n+D)$-round upper and lower bounds for some sort of approximation guarantee \cite{sarma2012distributed}. The notion of low-congestion shortcuts was invented as a framework  of circumventing these lower bounds \cite{ghaffari2016distributed}. Specifically, the ideas present in \cite{ghaffari2016distributed} can be turned into very short and clean $\tilde O(D+\sr n)$ round algorithms for general graphs, and near-optimal $\tilde O(D)$ round algorithms for special classes of graphs, for problems such as MST and Min-Cut.

However, the shortcut framework cannot be applied directly to the SSSP problem, since, unlike MST and Min-Cut, shortest path problems are not inherently parallelizable. For SSSP, a new technique based on multiplicative weights results in a $(1+\e)$-approximation to SSSP in $\tilde O(D+\sr n)$ time on general graphs \cite{becker2016near}. However, until this paper, not much work has been done on circumventing the $\tilde\Om(D+\sr n)$ lower bound on restricted classes of graphs or otherwise.

As a subroutine to computing shortest paths, we will be running low-diameter graph decompositions. Low diameter decompositions have a long history in the centralized~\cite{bartal1996probabilistic,linial1991decomposing} and parallel~\cite{awerbuch1992low,miller2013parallel,blelloch2014nearly} settings, and have been applied in the distributed setting to compute a network decomposition with low ``chromatic number''~\cite{elkin2016distributed}. 

\section{Distance-Preserving Tree}

Let $G$ be a weighted graph with $Q_G$-quality shortcuts. For a reader not familiar with shortcuts or the material in Appendix~\ref{sec:shortcutframework} of the full version, the parameter $Q_G$ intuitively measures how easy it is for connected components of $G$ to communicate within each other.
As a general rule, the ``nicer'' the graph $G$ is, the smaller the quantity $Q_G$ and the closer it gets to the optimal $D$. For example, if $G$ is a planar graph, then $Q_G=\tO(D)$.

We first consider the problem of finding a tree such that, for every pair of vertices $x,y\in V$, their distance is well-approximated with constant probability. 
Our algorithm is an adaptation of the algorithm of Section 5.4 from~\cite{alon1995graph}. 


To motivate the ideas behind the algorithm, we describe it in a parallel framework with graph contraction support. In each iteration, the algorithm runs a low diameter decomposition (defined below; see Appendix~\ref{sec:ldd} of the full version for details) on the graph and contracts each component into a single vertex.
To compute the tree as described above, take the set of edges inside the BFS trees formed by each LDD, and map them back to the original graph. The resulting tree is simply the (disjoint) union of these edges over all iterations. Of course, in a distributed framework, we cannot maintain contracted graphs, so we substitute each contracted vertex with a part of the original graph with zero-weight edges inside. To communicate efficiently between the parts, we establish shortcuts within each part.

\begin{defn}
For a weighted graph $G=(V,E)$, a low-diameter decomposition (LDD) of $G$ is a probabilistic distribution over partitions of $V$ into connected components $S_1,\lds,S_k$, such that 
\be
\im W.h.p., every induced graph $G[S_i]$ has low weighted diameter.
\im For every two vertices $x,y\in V$, the probability that they belong to the same component is bounded from below by some function depending on $d_G(x,y)$.
\ee
\end{defn}

We now describe the algorithm in detail. For a weight function $w:E\to \R$, denote $G(w)$ to be the graph $G$ whose edges are reweighted according to $w$. The algorithm maintains a weight function $w:E\to\{0\}\cup[R,\poly(n)]$ on the set of edges, for a given value $R$. The zero-weight edges connect vertices within each component, while the threshold $R$ increases geometrically over time. With a larger threshold $R$, we can compute the LDD on $G(\f1Rw)$, allowing the LDD to travel farther in the same amount of time. If $R$ is large enough, this graph still has edge weights at least 1 in between components, so computing the LDD is feasible in a distributed manner.

In addition to $w$, the algorithm also maintains a forest $T$, which gets new edges every iteration until it results in the approximate shortest path tree. Consider the following {\tt{LDDSubroutine}}, which we apply iteratively to $w$ and $T$.

\begin{mdframed}[roundcorner=4pt, backgroundcolor=white]
        \textbf{Algorithm $(w',T')={\tt{LDDSubroutine}}(w,T,\bt,R)$}



Algorithm:
\begin{enumerate}[itemsep=-1mm]
\im Initially, set $w':=w$ and $T':=T$.
\im Consider $G_0(w)$, the subgraph of $G$ with only the edges $e$ with $w(e)=0$.
\im Let $H$ be the (multi-)graph with every connected component of $G_0(w)$ contracted to a single vertex. Denote $w_H$ as the function $w$ restricted to the edges in $H$. 
        \im Simulate a LDD on  $H(\f1Rw_H)$ with parameter $\f1\bt$ (see Appendix~\ref{sec:ldd} of the full version). The specifics are deferred to the next section.
\im For every edge in $H$ that is part of a BFS tree in the LDD, add that edge to $T'$. 
\im For every edge $e$ in $H$ completely inside a LDD component, set $w'(e):=0$.
\im For every other edge $e$ in $H$, set $w'(e):=w(e)+\f{c_1}\bt\log n$ (for large enough constant $c_1$).
\im Output $(w',T')$.
\ee

\end{mdframed}

\subsection{Correctness}

The following two lemmas bound the maximum weighted diameter of a component, and therefore also the running time of the subroutine, as well as the probability that two vertices close together belong to the same component. Their proofs are natural generalizations of those in~\cite{miller2013parallel} and appear in Appendix~\ref{sec:ldd} of the full version.

\begin{restatable}{lemma}{LDDDiam} \label{LDDDiam}
W.h.p., each component in {\tt{LowDiameterDecomposition}} has weighted diameter $O(\f1\bt \logn)$.
\end{restatable}

\begin{restatable}{lemma}{LDDProb} \label{LDDProb}
For vertices $u,v\in V$ of (weighted) distance $d$, the probability that $u$ and $v$ belong to the same component is $e^{-O(d\bt)}$.
\end{restatable}

We now describe in more detail how to simulate the LDD in $H(\f1Rw_H)$ in the desired running time. Observe that we cannot directly compute the LDD on the contracted graph, since the contracted vertices are actually entire parts with limited communication between them. However, we can apply shortcuts to communicate quickly within the parts, up to the quality of the shortcut.

\begin{lemma} \label{LDDSim}
The LDD on the contracted graph (step 4  of {\tt LDDSubroutine})   can be simulated with a $\tilde{O}(Q_G)$ multiplicative overhead in running time. In other words, if the LDD takes $d$ rounds, then it can be simulated in $\tilde{O}(Q_Gd)$ rounds in the network $G$.
\end{lemma}

\begin{proof}
Define the parts of $V$ to be the connected components of $G$, and compute a set of $\tilde{O}(Q_G)$-quality shortcuts, one for each part.  In every round of the LDD on  $H(\f1Rw_H)$, we perform two steps sequentially: one to traverse nonzero weight edges between parts, and one to flood through the zero weight edges within each part. To take care of the edges between parts, note that every such edge has weight at least 1, so we can send them directly through the network $G$. To flood  through the zero edges within each part, it suffices to compute the minimum time $t$ that is received by any vertex, and then broadcast the message ``$t$'' to the entire part. By routing through shortcuts, this can be done in $\tilde{O}(Q_G)$ time per partition. Overall, every round of the LDD is replaced by $\tilde{O}(Q_G)$ rounds in the network $G$, hence the multiplicative overhead.
\end{proof}

Together with Lemma \ref{LDDDiam}, we get a running time of $\tilde{O}(\f1\bt Q_G)$. 
\begin{defn}
Let $w:E\to\R$ be a weight function, and $T\subseteq G$ a forest. Define $G_0(w)$ to be the subgraph of $G$ with only the edges $e$ with $w(e)=0$.
Let $C_1,C_2,\lds$ of $G$ be the connected components of $G_0(w)$. We say that $(w,T)$ satisfies the \textbf{subroutine invariant with parameter $R$} if the following conditions hold:
\be[itemsep=-1mm]
\im The weighted diameter of each part $C_i$ using edge weights in $G$ is  at most $R$.
\im Every edge within a part $C_i$ has weight 0 in $w$.
\im Every edge between two parts $C_i,C_j$ has weight at least $R$ in $w$.
\im For all $x,y$ belonging to the same part $C_i$, $d_T(x,y)\le R$.
\im $T$ has a spanning tree within each part $C_i$, and no edges in between parts.
\ee
\end{defn}

\begin{lemma}\label{SubProps}
Fix parameter $\bt$. Suppose that the input $(w,T)$ to {\tt{LDDSubroutine}} satisfies the subroutine invariant with parameter $R$. Then, w.h.p., for large enough constants $c_1$ and $c_2$,
\bi
\im The output $(w',T')$ satisfies the subroutine invariant with parameter $(\f{c_1}\bt\logn)R$.
\im For all $x,y\in V$, $\E[d_{G(w')}(x,y)]\le (c_2\logn)d_{G(w)}(x,y)$.
\ei
\end{lemma}

\begin{proof}
Note that the following properties of the invariant follow immediately:
\be[itemsep=-1mm]
\im[2.]Every edge within a part $C_i'$ has weight 0 in $w'$.
\im[3.]Every edge between two parts $C_i',C_j'$ has weight at least $(\f{c_1}\bt\logn)R$ in $w'$.
\im[5.]$T'$ has a spanning tree within each part $C_i'$, and no edges in between parts.

\ee

To prove invariant (4), suppose that $x,y\in V$ are in the same $C_i'$. If they are also in the same $C_i$, then the property holds by the input guarantee. Otherwise, by Lemma~\ref{LDDDiam}, w.h.p. the parts containing $x$ and $y$ have distance $O(\f{1}\bt\logn)$ in the BFS tree on $H(\f1Rw_H)$, which means that there is a path in the BFS tree that travels through  $O(\f{1}\bt\logn)$ vertices in $H(\f1Rw_H)$. We consider the distance  through edges in $H(\f1Rw_H)$ and through vertices in $H(\f1Rw_H)$ (which are actually parts in $G$) separately. For the edges, the distance  is at most $O(\f1\bt\logn)R$ in $H$, and each of these edges has weight at least that in $G$, giving $O(\f1\bt\logn)R$ total distance. For the vertices, traversing through $T$ inside the $O(\f1\bt\logn)$ parts takes $O(R)$ distance each, by the input guarantee, and $O(\f1\bt\logn)R$ distance overall. Combining the two arguments proves (4) $d_{T'}(x,y)\le (\f{c_1}\bt\logn)R$. Note that (4) immediately implies that (1) the weighted diameter of each part $C_i'$ using edge weights in $G$ is  at most $(\f{c_1}\bt\logn)R$.

Finally, we prove that $\E[d_{G(w')}(x,y)]\le (c_2\logn)d_{G(w)}(x,y)$. If $x,y\in V$ are in the same $C_i'$, then their distance in $G(w')$ is zero and the claim follows. Otherwise, consider the shortest path in $H$, which is also the shortest path in $H(\f1Rw_H)$. By Lemma \ref{LDDProb}, every edge $e$ on this path has probability at most $1-e^{-O(w_e\bt)}=O(w(e)\bt)$ of being cut between two components, so the expected length is at most $O(w(e)\bt)\cd \f{c_1}\bt\logn=O(w(e)\logn)$. By linearity of expectation, the expected multiplicative increase of the path in $H(\f1Rw_H)$, and also in $G(w')$, is $O(\logn)$.

\end{proof}

\subsection{Algorithm Main Loop}

In this section, we apply {\tt{LDDSubroutine}} recursively with geometrically increasing values of $R$. We show that the resulting forest approximates distances in expectation.

\begin{mdframed}[roundcorner=4pt, backgroundcolor=white]
        \textbf{Algorithm $T={\tt{ExpectedSPForest}}(G,\bt,R_0)$}

Input:

\bi[itemsep=-1mm]
\im $G=(V,E)$, the network graph with edge weights in $[1,\poly(n)]$.
\im $\bt=(\logn)^{-\Om(1)}$, freely chosen.
\ei

Algorithm:
\be[itemsep=-1mm]
\im Initially, set
$R^{(0)}:=1$,
$T^{(0)}:=\emptyset$,
and $w^{(0)}$ to have the same edge weights as $G$.
\im For $t=1,2,\lds$, while $R<n^c$ for large enough $c$:
 \be
 \im $(w^{(t)},T^{(t)}):={\tt LDDSubroutine}(w^{(t-1)},T^{(t-1)},\bt,R^{(t-1)})$.
 \im Set $R^{(t)}:=(\f{c_1}\bt\logn)R^{(t-1)}$.
 \ee
\im Output the forest obtained on the last iteration.
\ee

\end{mdframed}

Note that  $T$ is not guaranteed to be a tree at the end of the algorithm, so distances within $T$ can be infinite. However, a simple induction with linearity of expectation shows that the expected increase in length behaves in a controlled way:

\begin{lemma}\label{TthStep}
Let $G$ be a network graph with edge weights in $ [1,\poly(n)]$, and let $\bt:= (\logn)^{-\Om(1)}$. On the $t$\textsuperscript{th} iteration of {\tt{ExpectedSPForest}}, for any two vertices $x,y\in V$, $\E[d_{G(w^{(t)})}(x,y)]\le (c_2\log n)^t d_G(x,y)$.
\end{lemma}

We now show that we get approximate shortest paths with constant probability.

\begin{lemma}\label{ConstProbSP}
        Let $G$ be a network graph with edge weights in $ [1,\poly(n)]$, and let $\bt:= (\logn)^{-\Om(1)}$. The algorithm {\tt{ExpectedSPForest}} runs in  $\tilde{O}(\f1\bt Q_G)$ rounds. Consider the output forest $T$, and fix any two vertices $x,y\in V$. Then, $d_T(x,y)\ge d_G(x,y)$ always\footnote{In particular, $d_T(x,y)=\infty$ if $x$ and $y$ are not in the same connected component in $T$}, and with constant probability, $d_{T}(x,y)\le O(\f1\bt d_G(x,y)^{O(\log\logn)/\log(1/\bt)})\cd d_G(x,y)$.
\end{lemma}

\begin{proof}
For the running time, there are $O(\f{\logn}{\log(1/\bt)})$ iterations of the LDD, each of which takes $\tilde{O}(Q_G)$ time.

For simpler notation, define $M:=d_G(x,y)$. Since every edge added to $T$ has weight at least the weight of that same edge in $G$, we clearly have $d_T(x,y)\ge M$. To prove the other bound on $d_T(x,y)$, consider any iteration $t$ such that $R^{(t)}\ge2(c_2\logn)^tM$. (We later argue that such an iteration $t$ must exist.) By Lemma \ref{TthStep} and Markov's inequality, $d_{\tilde{G}^{(t)}}(x,y)<R^{(t)}$ with probability at least $\f12$. If this occurs, then $x$ and $y$ cannot belong to different parts at iteration $t$, since the distance between parts is at least $R^{(t)}$. By the subroutine guarantee, $d_{T^{(t)}}(x,y) = O(\f1\bt\logn)R^{(t-1)}=O(R^{(t)})$, and since the edges of $T^{(t)}$ are preserved for the rest of the algorithm, $d_T(x,y)=O(R^{(t)})$ as well. Therefore, for this value of $t$, the approximation factor is $2(c_2\logn)^t$ with probability at least $\f12$. 

It remains to find the smallest satisfying $t$. The condition on $t$ is equivalent to $(\f{c_1}\bt\logn)^t\ge2(c_2\logn)^t M$, or $t\ge \lc \f{\log(2M)}{\log(c_1/c_2)+\log(1/\bt)}\rc$. For $t$ achieving equality, we get
$$R^{(t)}= \lp \f{c_1}\bt\logn \rp^t \le  \lp \f{c_1}\bt\logn \rp^{\f{\log(2M)}{\log(c_1/c_2)+\log(1/\bt)} + 1} = O \lp \f1\bt (2M)^{1+\f{O(\log\logn)}{\log(1/\bt)}} \logn\rp,$$
as desired.

Lastly, we show that such an iteration $t$ must exist. In particular, we show that the value of $t$ chosen above satisfies $R^{(t)}\le n^c$ for some large enough constant $c$ in the algorithm. Since $M=\poly(n)$ and $R=1/\poly(n)$, we have
$$ t= \left\lceil \f{\log(2M)}{\log(c_1/c_2)+\log(1/\bt)}\right\rceil  = O\lp\f{\logn}{\log(1/\bt)}\rp .$$
Therefore,
$$ R^{(t)}=\lp \f{c_1}\bt\logn\rp^tR_0=\lp \f{\logn}{\bt}\rp^{O\lp\f{\logn}{\log(1/\bt)}\rp}= \lp\f1\bt\rp^{O\lp\f{\logn}{\log(1/\bt)}\rp} \cd (\logn)^{O\lp\f{\logn}{\log(1/\bt)}\rp} =n^{O(1)}\cd n^{O(1)},$$
where the last equality uses the fact that $\bt=(\logn)^{-\Om(1)}\implies \log(1/\bt)=\Om(\log\logn)$. Therefore, $R^{(t)}\le n^c$ for large enough $c$.
\end{proof}


From the shortest path forest, we can also derive the distances to each vertex $v$ from a specified source $s$. Below is the algorithm, which runs in $\tilde O(\f1\bt Q_G)$ rounds.

\begin{mdframed}[roundcorner=4pt, backgroundcolor=white]
        \textbf{Algorithm ${\tt{ExpectedSPDistance}}(G,\bt,s)$}

\be[itemsep=-1mm]
\im Run ${\tt ExpectedSPForest}(G,\bt)$ to obtain forest $T$. Set $\tilde T$ to be the connected component of $T$ that contains the source $s$.
\im For all vertices $v\notin \tilde T$, set $d(s,v):=\infty$.
        \im Run {\tt AggregatePathToRoot} (see Appendix~\ref{sec:tree} of the full version) with $x_v=1$ for all $v\in\tilde T$ to determine the depth of each vertex in the tree $\tilde T$ rooted at $s$.
\im Every vertex $v\in \tilde T\bs\{s\}$ computes its parent in the rooted tree, which it can determine by finding the one neighbor with smaller depth.
\im For each $v\in\tilde T\bs\{s\}$, set $x_v$ to be the weight of the edge to its parent, and set $x_s:=0$. Run {\tt AggregatePathToRoot} on these values to determine $d(s,v)$ for $v\in T$.
\ee

\end{mdframed}

\section{Solving SSSP and Related Problems}
\subsection{SSSP Trees}
In this section, we describe an algorithm that outputs an approximate single source shortest path tree with source $s$. At a high level, to boost the probability that distances are well-approximated, we construct many randomized trees and take a collective ``best'' tree.

\begin{mdframed}[roundcorner=4pt, backgroundcolor=white]
        \textbf{Algorithm ${\tt{SSSPTree}}(G,\bt,s)$}
\be[itemsep=-1mm]
\im Repeat  ${\tt ExpectedSPDistance}(G,\bt,s)$ $\Th(\logn)$ times to obtain distances $d_{T_i}(v):=d_{T_i}(s,v)$.
\im For each vertex $v$, set $d_{\min}(v):=\min_id_{T_i}(v)$.
\im For each vertex $v$ except the source, connect an edge to some neighbor $u$ that satisfies $d_{\min}(u)+w_{(u,v)}\le d_{\min}(v)$. Return the tree $T^*$ of all such edges.
\ee
\end{mdframed}

\begin{lemma} \label{SSSPBeta}
Let $G$ be a network graph with edge weights in $ [1,\poly(n)]$, and let $\bt:= (\logn)^{-\Om(1)}$. W.h.p., {\tt SSSPTree} runs  for  $\tilde{O}(\f1\bt Q_G)$ rounds and outputs a shortest path tree that $O(\f1\bt d_G(v)^{\f{O(\log\logn)}{\log(1/\bt)}}\logn)$-approximates distances from the source to each $v$.
\end{lemma}

\begin{proof}
Observe that in step 3 of {\tt SSSPTree}, such a neighbor always exists, since in the tree $T_i$ that achieves distance $d_{\min}(v)$ to $v$, the parent $u$ of $v$ in $T_i$ satisfies $d_{\min}(u)+w_{(u,v)}=d_{\min}(v)$. To show that $d_{T^*}(v)\le d(v)$ for each $v$, consider the path  $s=v_0,v_1,v_2,\lds,v_\el=v$ in $T^*$. We have $w(v_i,v_{i-1})\le d_{\min}(v_i)-d_{\min}(v_{i-1})$ for each $i$, and summing up the inequalities gives the result.

From Lemma \ref{ConstProbSP}, each vertex $v$ achieves the desired approximation with constant probability. By taking the minimum $d_{T_i}(v)$ over $\Th(\logn)$ trees, this approximation is satisfied w.h.p. for every $v$, giving $d_{T^*}(v)\le d_{\min}(v)=O(\f1\bt d_G(v)^{1+\f{O(\log\logn)}{\log(1/\bt)}}\logn)\cd d_G(v)$.
\end{proof}

This concludes Theorem~\ref{SSSP}, restated below.


\SSSP*

\subsection{Distance Labeling Schemes}

We restate our main result on approximate distance labeling schemes.

\DistLabel*

\begin{proof}
For each $t$ from $1$ to $\lc \log(c_1/c_2)+\log(1/\bt) \rc$, run  {\tt{ExpectedSPForest}} $\Th(\logn)$ times with $R_0:=2^{-t}$. By analysis from Lemma \ref{ConstProbSP} and Theorem \ref{SSSP}, w.h.p., for every $x,y\in V$, there is an iteration of   {\tt{ExpectedSPForest}} with  $R=O(d_G(v)^{1+\f{O(\log\logn)}{\log(1/\bt)}})$ that outputs a cluster containing both $x$ and $y$. The total number of rounds is $\tilde{O}(\f1\bt Q_G)$.

In each of the $O(\log^2n)$ iterations of  {\tt{ExpectedSPForest}}, consider all of the clusters formed throughout the algorithm, and give each one a unique ID. For every iteration with parameter $R$ and a cluster formed in that iteration, assign  to every vertex within the cluster the label $(\text{ID}, R)$.  Each vertex is assigned to $O(\f\logn{\log(1/\bt)})$ clusters per  {\tt{ExpectedSPForest}}, so the label size is $\polylog(n)$.

To compute distances given two vertices $x,y\in V$, simply output the minimum possible $R$ over all clusters that contain both $x$ and $y$, which is easily computed with the labels of $x$ and $y$. By the analysis above, the minimum possible $R$ gives the desired approximation factor $O(d_G(v)^{{O(\log\logn)}/{\log(1/\bt)}})=O(n^{{O(\log\logn)}/{\log(1/\bt)}})$. 
\end{proof}

\subsection{Transshipment Problem
}

Let $G$ be a transshipment network with demand $d_v$ at each node $v$. The following algorithm computes an approximate transshipment flow in expectation.

\begin{mdframed}[roundcorner=4pt, backgroundcolor=white]
        \textbf{Algorithm {\tt{ExpectedTS}}}
\be[itemsep=-1mm]
\im Run  {\tt ExpectedSPForest} and root the tree $T$ arbitrarily.
        \im Using  {\tt AggregateSubtree} (see Appendix~\ref{sec:tree}), compute $F(v):=\Sum_{u\in S_v}d_v$ for all $v$, where $S_v$ is the subtree rooted at $v$.
\im For each edge $(v,p)\in T$ with $p$ the parent of $v$ in the rooted tree, direct $F(v)$ flow from $v$ to $p$. (If $F(v)$ is negative, then direct the flow the other way.)
\ee
\end{mdframed}

\begin{lemma} \label{ExpectedTSCost}
Let $G$ be a network graph with edge weights in $ [1,\poly(n)]$ and demands that sum to zero, and let $\bt:= (\logn)^{-\Om(1)}$. The expected total cost of {\tt{ExpectedTS}} is within $\tilde{O}(\f1\bt n^{O(\log\logn)/\log(1/\bt)})$ of optimum.
\end{lemma}

\begin{proof}
Decompose the optimal solution into a set of (shortest) paths. For a path from $s$ to $t$, we have $\E[d_T(s,t)]=\tilde O(\f1\bt n^{O(\log\logn)/\log(1/\bt)})\cd d_G(s,t)$ by Lemma \ref{ConstProbSP}, and by linearity of expectation, the cost $C$ of routing each of these paths through $T$ gives an expected $\tilde O(\f1\bt n^{O(\log\logn)/\log(1/\bt)})$ approximation. It remains to show that the total cost of {\tt{ExpectedTS}} is at most $C$. If {\tt{ExpectedTS}} places $F(e)$ flow along an edge $e$, then the total demand difference between the two halves of the tree split at $e$ is $|2F|$. Therefore, any sequence of paths along $T$ that satisfies all demands must route at least $|F|$ flow along edge $e$. It follows that $C$ must be at least the cost of {\tt{ExpectedTS}}.
\end{proof}
By running {\tt ExpectedTS} repeatedly and taking the overall best flow, we obtain our main result for transshipment.

\TS*

\begin{proof}
Run  {\tt ExpectedTS} $\Th(\logn)$ many times and output the minimum total cost. By Markov's inequality and Lemma \ref{ExpectedTSCost}, w.h.p., some iteration achieves within twice the expected approximation of $\tilde O(\f1\bt n^{O(\log\logn)/\log(1/\bt)})$.
\end{proof}

\section{Conclusion and Future Work}

Using the shortcuts framework from \cite{ghaffari2016distributed,haeupler2016low}, we give the first nontrivial approximation algorithms for shortest path problems which run in $o(\sr n+D)$ time on non-pathological network topologies. Our algorithms feature a tuneable parameter $\bt$ that represents the balance between approximation ratio and running time. For certain values of $\bt$, we obtain polylogarithmic-approximate solutions in $\tilde O(n^\e \cdot Q_G)$ rounds for the shortest path and distance labeling problems. While sublogarithmic approximation ratios are known to be impossible (even existentially) for labeling schemes with polylogarithmic labels we believe that our approximation guarantees can likely be improved for nice family of graphs, and, in the case of the SSSP-tree and transshipment problems, even generally. 

In particular, for the quite general set of minor closed families of graphs one might be able to use more sophisticated low-diameter decompositions, such as \cite{abraham2014cops}, which would directly lead to $O(1)$-approximation guarantees for such networks in our framework. However, \cite{abraham2014cops} is written for the sequential setting and making the algorithms in \cite{abraham2014cops} distributed and compatible with the shortcut framework is a nontrivial extension, which we plan to explore for the journal version of this work. 

More importantly, it seems possible that our non-trivial approximation ratios for the SSSP-tree and transshipment problem can be improved all the way to $(1+\epsilon)$-approximations using tools from continuous optimization, such as, gradient descent or the multiplicative weights method. As one example, the recent and brilliant work of Becker et al. \cite{becker2016near} shows how to obtain a $(1+\epsilon)$-approximation for the SSSP-tree problem and the transshipment problem by computing $\tO(\alpha^2)$ many $\alpha$-approximations to the transshipment problem. This work also demonstrates that the required updates to weight and demand vectors can be performed in various non-centralized models, including CONGEST. If this method could be applied to our transshipment algorithm, we could choose $\beta = 2^{-O(\sqrt{\log n \log \log n})}$ to get a $2^{O(\sqrt{\log n \log \log n})}$-approximate solution to the transshipment problem in $Q_G \cdot 2^{O(\sqrt{\log n \log \log n})}$ rounds, which could then be transformed into a $(1+\epsilon)$ approximation with the exact same running time (up to the constant hidden by the $O$-notation). This extension is highly nontrivial as well and left for future work. 
  
 %
%



\newpage
\bibliographystyle{abbrv}
\bibliography{refs}

\begin{thebibliography}{10}

\bibitem{abraham2014cops}
I.~Abraham, C.~Gavoille, A.~Gupta, O.~Neiman, and K.~Talwar.
\newblock Cops, robbers, and threatening skeletons: Padded decomposition for
  minor-free graphs.
\newblock In {\em Proceedings of the 46th Annual ACM Symposium on Theory of
  Computing}, pages 79--88. ACM, 2014.

\bibitem{alon1995graph}
N.~Alon, R.~M. Karp, D.~Peleg, and D.~West.
\newblock A graph-theoretic game and its application to the k-server problem.
\newblock {\em SIAM Journal on Computing}, 24(1):78--100, 1995.

\bibitem{awerbuch1992low}
B.~Awerbuch, B.~Berger, L.~Cowen, and D.~Peleg.
\newblock Low-diameter graph decomposition is in nc.
\newblock In {\em Scandinavian Workshop on Algorithm Theory}, pages 83--93.
  Springer, 1992.

\bibitem{bartal1996probabilistic}
Y.~Bartal.
\newblock Probabilistic approximation of metric spaces and its algorithmic
  applications.
\newblock In {\em Foundations of Computer Science, 1996. Proceedings., 37th
  Annual Symposium on}, pages 184--193. IEEE, 1996.

\bibitem{becker2016near}
R.~Becker, A.~Karrenbauer, S.~Krinninger, and C.~Lenzen.
\newblock Near-optimal approximate shortest paths and transshipment in
  distributed and streaming models.
\newblock In {\em International Symposium on Distributed Computing}, 2017.

\bibitem{blelloch2014nearly}
G.~E. Blelloch, A.~Gupta, I.~Koutis, G.~L. Miller, R.~Peng, and K.~Tangwongsan.
\newblock Nearly-linear work parallel sdd solvers, low-diameter decomposition,
  and low-stretch subgraphs.
\newblock {\em Theory of Computing Systems}, 55(3):521--554, 2014.

\bibitem{elkin2016distributed}
M.~Elkin and O.~Neiman.
\newblock Distributed strong diameter network decomposition.
\newblock In {\em Proceedings of the 2016 ACM Symposium on Principles of
  Distributed Computing}, pages 211--216. ACM, 2016.

\bibitem{ghaffari2016distributed}
M.~Ghaffari and B.~Haeupler.
\newblock Distributed algorithms for planar networks ii: Low-congestion
  shortcuts, mst, and min-cut.
\newblock In {\em Proceedings of the Twenty-Seventh Annual ACM-SIAM Symposium
  on Discrete Algorithms}, pages 202--219. Society for Industrial and Applied
  Mathematics, 2016.

\bibitem{haeupler2016low}
B.~Haeupler, T.~Izumi, and G.~Zuzic.
\newblock Low-congestion shortcuts without embedding.
\newblock In {\em Proceedings of the 2016 ACM Symposium on Principles of
  Distributed Computing}, pages 451--460. ACM, 2016.

\bibitem{haeupler2016near}
B.~Haeupler, T.~Izumi, and G.~Zuzic.
\newblock Near-optimal low-congestion shortcuts on bounded parameter graphs.
\newblock In {\em International Symposium on Distributed Computing}, pages
  158--172. Springer, 2016.

\bibitem{minorspersonalcommunication}
B.~Haeupler, G.~Zuzic, and J.~Li.
\newblock Low-congestion shortcuts for any minor closed family.
\newblock In {\em personal communications}, 2017.

\bibitem{henzinger16almost}
M.~Henzinger, S.~Krinninger, and D.~Nanongkai.
\newblock An almost-tight distributed algorithm for computing single-source
  shortest paths.
\newblock In {\em Proceedings of the ACM Symposium on Theory of Computing},
  2016.

\bibitem{kutten1995fast}
S.~Kutten and D.~Peleg.
\newblock Fast distributed construction of k-dominating sets and applications.
\newblock In {\em Proceedings of the fourteenth annual ACM symposium on
  Principles of distributed computing}, pages 238--251. ACM, 1995.

\bibitem{lenzen2015fast}
C.~Lenzen and B.~Patt-Shamir.
\newblock Fast partial distance estimation and applications.
\newblock In {\em Proceedings of the 2015 ACM Symposium on Principles of
  Distributed Computing}, pages 153--162. ACM, 2015.

\bibitem{linial1991decomposing}
N.~Linial and M.~E. Saks.
\newblock Decomposing graphs into regions of small diameter.
\newblock In {\em SODA}, volume~91, pages 320--330, 1991.

\bibitem{miller2013parallel}
G.~L. Miller, R.~Peng, and S.~C. Xu.
\newblock Parallel graph decompositions using random shifts.
\newblock In {\em Proceedings of the twenty-fifth annual ACM symposium on
  Parallelism in algorithms and architectures}, pages 196--203. ACM, 2013.

\bibitem{nanongkai2014distributed}
D.~Nanongkai.
\newblock Distributed approximation algorithms for weighted shortest paths.
\newblock In {\em Proceedings of the ACM Symposium on Theory of Computing},
  pages 565--573, 2014.

\bibitem{nanongkai2014almost}
D.~Nanongkai and H.-H. Su.
\newblock Almost-tight distributed minimum cut algorithms.
\newblock In {\em International Symposium on Distributed Computing}, pages
  439--453. Springer, 2014.

\bibitem{sarma2012distributed}
A.~D. Sarma, S.~Holzer, L.~Kor, A.~Korman, D.~Nanongkai, G.~Pandurangan,
  D.~Peleg, and R.~Wattenhofer.
\newblock Distributed verification and hardness of distributed approximation.
\newblock {\em SIAM Journal on Computing}, 41(5):1235--1265, 2012.

\end{thebibliography}

\appendix

\newpage

\begin{center}\Huge Appendices\end{center}

\section{Low-Congestion Shortcut Framework and Beating $\tODn$}\label{sec:shortcutframework}

This section provides a more detailed explanation of the shortcut framework \cite{ghaffari2016distributed} and the powerful tools it provides, obtaining faster algorithms that are much tighter coupled to the fastest algorithm achievable in a given topology.

\subsection{Beating $\tODn$ and (Instance) Optimality}

In TCS, running times of algorithms are typically measured for worst-case inputs and expressed as asymptotic functions of the instance size, i.e., the number of nodes $n$ for graph problems. 
This approach is also standard in distributed computing when studying running times and complexities of local graph/network problems, such as in various coloring and decomposition problems.
However, many interesting optimization algorithms are non-local and trivially require a running time which is at least as large as the network diameter $D$. This means that in a pathological worst-case topology, such as a line network, $\Omega(n)$ running times are required. Technically, this makes any $O(n)$ round algorithm optimal, in the sense that no better running time, when measured only in terms of $n$, can be achieved for every network. Furthermore obtaining such an ``optimal'' $O(n)$ round algorithms is often essentially trivial for many problems of interest, including the shortest path problem. However, such algorithms are far from satisfactory, and it would have been a great loss to the field of distributed computing if theoreticians would have stopped trying to obtain algorithms that are ``faster'' than the trivial but ``optimal'' $O(n)$ round algorithms. 

\smallskip

In particular, in any given application of distributed computing, it is quite plausible that many instances of a distributed problem need to be solved (e.g., as subroutines) on a given fixed topology. Therefore, it is likely that one might encounter a worst-case instance. The network topology itself, however, is typically fixed and  most likely not of pathological $\Omega(n)$ diameter. The goal is thus to compute a solution on the given topology as fast as possible. The assurance that one should be happy with an optimal running time of $\Theta(n),$ because in some completely different pathological line network topology no better running time can be obtained, is not very helpful and strong given that the  network topology of interest is unlikely to be of this type, and thus typically allows for a drastically faster running time. 

\smallskip

With this in mind,  the distributed computing community is employing a finer way to analyze  non-local distributed optimization algorithms by expressing their running times in terms of $n$ and $D$. The ultimate goal would be to achieve running times of $\tilde{O}(D)$. We remark that the optimality (up to logarithmic factors) of such a complexity is qualitatively very different, due to the trivial $\Omega(D)$ lower bound, which holds not just for some network of diameter $D,$ but for \textit{any} network of diameter $D$. In particular, an $\tilde{O}(D)$ algorithm for a non-local problem is \emph{instance optimal}, i.e, the running time of the algorithm on any given topology is as fast as it can be on \emph{this} (and not some other) topology.

\smallskip

Progress on distributed lower bounds in the last decade, however, has made it very clear that such $\tilde{O}(D)$ algorithms cannot exist in general. In fact, the well-known lower bound framework of \cite{sarma2012distributed} give the wide-ranging and devastating result that even the simplest non-local optimization problem, even if one merely wants so barely-non-trivial approximation guarantees, cannot be obtained in less than $\tilde{O}(\sqrt{n} + D)$ rounds in general. In particular, there exists pathological network topologies with a tiny, say (poly)logarithmic, diameter on which any optimization algorithm requires not just polylogarithic but $\Omega(\sqrt{n})$ many rounds. At the same time, much progress has been made on the algorithmic side as well. Many celebrated results with beautiful and highly sophisticated algorithms now achieve the $\tilde{\Omega}(\sqrt{n} + D)$ running time for many problems, including the SSSP problem. 


\smallskip

Unfortunately, these algorithms are far from being instance optimal, and in fact their running times inherently remain $\tilde{\Theta}(\sqrt{n}+D)$ for \emph{any} topology.
%
%
 Even worse, the structure of the topology given in the lower bound of \cite{sarma2012distributed} is (topologically) very complicated\footnote{One way to make this assertion formal is to note that cannot be embedded into any surface unless it has a huge $\Omega(n)$ genus(!) and it contains a complete bipartite $\sqrt{n} \times \sqrt{n}$ (induced) minor.} and unlikely to occur in practical networks or in other network topologies of interest.
To the contrary, on many classical families of networks that practitioners and/or theoreticians consider interesting, the lower bound does provably not apply~\cite{haeupler2016low,ghaffari2016distributed}. Overall, current distributed shortest path algorithms are optimal in these sense that their running time cannot be improved as a function of $n$ and $D$, because there exists a pathological network topology where one cannot do better. There is, however, the distinct possibility that for networks of interest in which one wants to run these algorithms, much faster running times are possible than the $\Omega(\sqrt{n})$ rounds taken by current algorithms. 
 
\smallskip

This is reminiscent of the problems in using only functions of $n$ as a measure of complexity. However, a fix is not as immediate, given that the $\tilde{\Omega}(\sqrt{n})$ lower bound is much more intricate than the trivial $\Omega(D)$ lower bound. In particular, it is  much less clear what characteristics generally make a topology ``hard'' and how this ``hardness'' can be meaningfully defined, characterized, or parameterized. One way to circumvent this issue is to directly look at classes of network topologies which are (arguably) of interest, such as planar networks, or networks with small tree-width. The problem with this approach, however, is the somewhat limited scope of such a direction. Furthermore, algorithms specifically aimed at a certain class of networks are sometimes not be very robust, in that they might fail completely if the structural assumption is just slightly violated. Given that practical networks will likely not exactly fit into one of the presumed graph classes, such algorithms are less desirable. 


The shortcut framework of \cite{ghaffari2016distributed,haeupler2016low}, which is discussed next, addresses these issues and gives a very general and powerful solution.

\subsection{The Low-Congestion Shortcut Framework}

The Low-Congestion Shortcut Framework of \cite{ghaffari2016distributed,haeupler2016low} was designed to capture the essence of what makes the pathological topology in the $\tOm(\sqn)$ lower bound hard. Shortcuts do this, however, in a way that allows them to be used as a powerful \emph{algorithmic} tool in any topology which does not have such characteristics. The framework furthermore allows one to define a shortcut quality parameter for any topology, which essentially captures how hard or easy it is to route information for this given network. 

\smallskip

Shortcuts are defined with respect to a collection of disjoint vertex subsets $S_1, \ldots, S_N \subseteq V$, each being connected. We call these subsets \emph{parts} and speak of the collection as a \emph{(valid) partition}. We note that not every node in $V$ must be in a part. 

\smallskip

The way one should think of a part is as a distributed sub-problem in which one wants to perform some simple communication. In particular, the goal is it to compute a simple aggregate function within each subset. 

\smallskip

This problem arises naturally in many settings, such as Boruvka's MST algorithm. For another example related to SSSP, consider an SSSP instance in which several edge weights are zero, maybe because they have been rounded down. The zero edges now induce several connected subsets which will be exactly our parts. In a shortest path algorithm all nodes in a part need to have the same SSSP distance. Even if we just want to verify some given SSSP distances approximately each node in such a part must essentially learn whether the minimum SSSP distance assigned to a node in its part is the same (or much different) from its own supposed SSSP distance. 

\smallskip

Of course, it is easy to compute the minimum value in each part by having all nodes in a part flood the minimum value seen so far to all its neighbors, in time equal to the strong diameter of each part. Unfortunately, however, the strong diameter of a subset of vertices can be much much larger than the weak diameter or the diameter of the underlying graph. This necessitates communicating the information of a part via other edges in the graph. However, if too many parts try to communicate their information using the same edge in the graph, they cause congestion on this edge. One can try to minimize congestion by routing information via different possibly slightly longer paths. Depending on the topology, there is then a tradeoff between the longest path along which information of a part is routed, which we call dilation, and the maximum congestion caused along an edge. Shortcuts exactly capture which tradeoffs between dilation and congestion can be achieved for a given topology and valid partition. 

\begin{defn}[Shortcuts] Given a graph $G=(V, E)$ and a valid partition $S_1$, \dots, $S_{N} \subset V$ a \textbf{$c$-congestion $d$-dilation} shortcut specifies a shortcut edge set $E_i \subseteq E$ for each $S_i$ such that:
\begin{itemize} 
\item[(\textbf{1})] For each $i$, the diameter of the subgraph $G[S_i]+E_i$ is at most $d$.  
\item[(\textbf{2})] Each edge $e\in E$ is contained in at most $c$ shortcut edge sets.
\end{itemize}
\end{defn}

It is intuitively clear that if there is a distributed algorithm with round complexity $T$ that achieves communication in every part in parallel then tracing the way the information has flown results in shortcuts of dilation at most $T$ and congestion $\tO(T)$, because information cannot travel faster than one hop per round and even if each part sends merely on bit of information along some edge used by it, at most $T/\log n$ parts can use an edge because the total capacity of an edge in $T$ rounds is at most $T\log n$. In this way the existence of a $\tO(T)$ congestion $T$ dilation shortcut seems essentially a necessary requirement for a fast algorithm solving even the most basic part-wise communication problem to be able to succeed.

\smallskip

The real power of shortcuts, however, is that the opposite direction holds true. In particular, when given a shortcut a simple distributed algorithm computes any simple aggregate function (such as min,  sum, xor, and, etc.) of all nodes in each $S_i$ set in parallel using only $\tO(c+d)$ time:

\begin{lemma}\label{lem:partwisecom}
Suppose we have a $d$-dilation $c$-congestion shortcut for a valid partition $S_1,\lds,S_N$ and suppose each node $v \in V$ has a value $x_v$ of logarithmic bit size. Now, let $\oplus$ be a commutative function. Then, there is a simple $\tilde O(c + d)$ round distributed algorithm which computes at each node $v$ in a part $S$ the value $\displaystyle\bigoplus_{v\in S}x_v$.
\end{lemma}

This essentially shows that whether or not it is possible to solve the part-wise communication problem in $\tO(T)$ rounds depends only on whether or not the topology supports a shortcut with $T = \tO(c+d)$. Hence, it makes sense to define the \emph{quality} $Q$ of a shortcut to be the sum of its congestion and dilation. 

\begin{defn}
The \textbf{quality} $Q=Q(\mathcal E)$ of a shortcut $\mathcal E$ is equal to $c+d$, where $c$ and $d$ are the dilation and congestion of the shortcut, respectively.
\end{defn}

\smallskip

While in general one should not hope for a dilation and thus quality better than the network diameter, we note that an easy argument shows that {\bfseries any valid partitioning in any graph has a shortcut of quality $\mathbf{Q = \sqn + D}$}, where $D$ is the network diameter~\cite{ghaffari2016distributed}: simply give any part consisting of at least $\sqn$ nodes all edges (or the edges of any BFS-tree) in its shortcut set, and leave the shortcuts sets of smaller parts empty. Now, the dilation of all large parts becomes $D$ while the congestion of small parts can be at most $\sqn$. The congestion on any edge is furthermore at most the number of large parts, of which there can be at most $\sqn$. It is also easy to see that the lower bound topology is exactly designed to either force a dilation or congestion of $\tOm(\sqn)$ leading the best shortcut (for the natural partitioning) to have absolute worst possible quality (for a low diameter network) of $\tilde{\Theta}(\sqn)$. On the other hand \cite{ghaffari2016distributed} and \cite{haeupler2016near} show that any partitioning in planar, bounded or polylogarithmic genus topologies, and bounded or polylogarithmic treewidth or pathwidth topologies has a shortcut of quality $\tO(D)$. All these shortcuts furthermore have the nice property that they are tree restricted, i.e., the union of the shortcut edge sets can be chosen to come from any (low diameter) tree one chooses, such as a BFS tree.

\begin{theorem}
Given a graph $G=(V, E)$  of polylogarithmic genus or treewidth, a valid partition $S_1$, \dots, $S_{N} \subset V$, and a spanning tree $T\subseteq E$ of $G$ of diameter $O(D)$, a \textbf{$T$-restricted} shortcut is a shortcut $\{E_i \subseteq E: i=1,\lds,N\}$ with the additional property that $E_i \subseteq T$. We say that a shortcut is \textbf{tree-restricted} if it is $T$-restricted for some spanning tree $T$ of $G$ of diameter $O(D)$.
\end{theorem}

\smallskip

The final surprising and crucial ingredient for the framework is a simple and efficient distributed algorithm which, for any topology and any partition and any low diameter tree, computes a polylogarithmic approximation to the best possible tree restricted shortcut for this topology and partition~\cite{haeupler2016low}. If the best such shortcut is of quality $Q$, then this algorithm runs in $\tO(Q)$ time. This allows any distributed algorithm to
construct an approximately optimal shortcut on the fly, and then perform a part-wise communication primitive from Lemma~\ref{lem:partwisecom}, while overall not taking longer than the $\tO(Q)$ one should expect to take for such a communication anyway. Differently speaking, algorithms using shortcuts automatically adjust and have a running time related to the best shortcuts possible in a given topology. 

\smallskip

To accurately specify the running time of our algorithms, 
we define the quality $Q_G$ of a network $G$ to be the best possible tree-restricted shortcut quality that can be achieved for any valid partition (and its BFS tree).

\begin{defn}
Given a graph $G$, the shortcut quality $Q_G$ of $G$ is defined as
$$Q_G = \max_{\substack{\mathcal S=\{S_1,\lds,S_n\}\\\text{valid partition}}} \min_{\substack{\text{tree-restricted}\\\text{shortcut }\mathcal E\\ \text{for }\mathcal S}} Q(\mathcal E) $$
\end{defn}

 This allows us to
use the part-wise communication procedure in our algorithms and express the final running times in terms of $Q_G$. We note that taking the worst-case over all partitions is merely for clarity and that our $\tO(Q_G)$ and $O(Q_G n^{\eps})$ algorithms still adjust to the actual input as well, meaning they will run faster on easy inputs even if $Q_G$ is large, certifying that a harder input could have been embedded into the same topology. 

\smallskip

In summary, for any network with topology $G$, the quantity $Q_G$ captures well how hard it is to solve the very simple part-wise communication problem described in Lemma~\ref{lem:partwisecom} in the CONGEST model. Furthermore, there are simple distributed algorithms which solve the part-wise communication problem in $\tO(Q_G)$ rounds. In the worst-case, this hardness, and thus also running time, is $\tilde{\Theta}(\sqn+D)$ for a network with $n$ nodes and diameter $D$~\cite{ghaffari2016distributed}. In particular, the  topology used for the lower bound in \cite{sarma2012distributed} is pathological and has the worst-case hardness among all topologies (with low diameters). In many other networks of interest, however, the hardness $Q_G$ is much lower and in fact only $\tO(D)$. While $\tO(Q_G)$ round shortcut-based algorithms necessarily have a worst-case running time of $\tO(\sqn + D)$ when expressed in terms of $n$ and $D$, they are essentially\footnote{The only reason why this is only essentially and not fully formally true is due to the possible difference between the best tree restricted and best possible shortcut (which does not exist in all network families studied in \cite{haeupler2016near} and \cite{ghaffari2016distributed}) and due to the fact that the necessity of a good shortcut for a fast solution of the part-wise communication problem seems very intuitive but cannot be easily formalized --- see above and also \cite{ghaffari2016distributed} for more details.} running as fast as the given topology (and to some extent even the given input) allows it.

\section{Tree Algorithms}\label{sec:tree}

In this section, we consider algorithms on trees in a distributed setting. Let $G=(V,E)$ be a distributed network with $Q_G$-quality shortcuts, and let $T\subseteq G$ be an embedded tree. Note that the diameter of $G$ can be much smaller than that of $T$, and we want our algorithms to have performance dependent on the diameter of $G$. 
\subsection{Heads/Tails Clustering Algorithm}

In this section, we present the Heads/Tails low-diameter hierarchical clustering algorithm in a distributed setting. This section highlights our first use of shortcuts and provides the intuition behind associating shortcuts with graph contractions.

\begin{defn}
        Let $G=(V,E)$ be a connected graph. A \textbf{low-diameter  hierarchical clustering} of $G$ consists of a sequence of partitions $\mc P_0,\mc P_1,\lds,\mc P_K$ of $V$, with $\mc P_i=\{C_{i,1},C_{i,2},\lds\}$ ($C_{i,j}\subseteq V$), that satisfies the following:
\be
\im Every $C_{i,j}$ is connected.
\im $\mc P_0=\{\{v_1\},\{v_2\},\lds\}$ is the partition into singleton vertices, and $\mc P_K=\{V\}$ partitions $V$ into a single cluster.
\im For each cluster $C_{i,j}$ ($i>1$), there exists a set of clusters $C_{i-1,k_1},C_{i-1,k_2},\lds$ whose (disjoint) union is $C_{i,j}$. Moreover, if we take the graph induced on $C_{i,j}$ and contract clusters $C_{i-1,k_1},C_{i-1,k_2},\lds$, then the resulting graph has diameter $O(1)$.
\ee
\end{defn}

We first describe the Heads/Tails algorithm in the parallel model with graph contractions, and then highlight the technical differences needed in a distributed framework.

In the parallel model, we can view the clusters as the vertices in a \textit{contracted} graph, where every vertex represents the contraction of a cluster of the original graph. We highlight the algorithm in the contraction below, with an example shown in Figure \ref{fig:HTParallel}.
\begin{mdframed}[roundcorner=4pt, backgroundcolor=white]
\textbf{Algorithm {\tt{HeadsTailsParallel}}}

\bi
\im For $i=1,2,\lds$:
\be
\im Every vertex flips either Heads or Tails.
\im For each vertex that flips Tails and which has at least one Heads neighbor, connect the vertex with an arbitrarily chosen Heads neighbor.
\im Contract the connected components.
\im Number the remaining vertices from $1$ to $k$, and output the partition $\mc P_i:=\{C_1,\lds,C_k\}$, where $C_j$ is the set of \textit{original} vertices contracted to vertex $j$.
\im Repeat until there are no more edges.
\ee
\ei
\end{mdframed}

\begin{lemma} \label{HTParallel}
{\tt{HeadsTailsParallel}} terminates in $O(\logn)$ rounds w.h.p., and outputs a low-diameter hierarchical clustering.
\end{lemma}

\begin{figure}
\centering
\begin{tikzpicture}[scale=.9, transform shape]

\begin{scope}
\tikzstyle{every node} = [draw,shape=circle]
\node (1) at (2,4) {$a^H$};
\node (2) at (1,2) {$b^T$};
\node (3) at (0,0) {$c^T$};
\node (4) at (1,0) {$d^H$};
\node (5) at (2,0) {$e^H$};
\node (6) at (3,2) {$f^T$};
\tikzstyle{every node} = [];
\foreach \from/\to in {1/2,2/3,2/5}
  \draw [] (\from) -- (\to);
\foreach \from/\to in {2/4,6/1}
  \draw [line width=2pt, ->] (\from) -- (\to);
\end{scope}

\node at (4.5,2) {\Huge$\Ra$};

\begin{scope}[shift={(5,0)}]
\tikzstyle{every node} = [draw,shape=circle]
\node (1) at (2,4) {$af^T$};
\node (2) at (1,2) {$bd^H$};
\node (3) at (0,0) {$c^H$};
\node (5) at (2,0) {$e^T$};
\foreach \from/\to in {2/3}
  \draw [] (\from) -- (\to);
\foreach \from/\to in {1/2,5/2}
  \draw [line width=2pt, ->] (\from) -- (\to);
\end{scope}

\node at (8.5,2) {\Huge$\Ra$};

\begin{scope}[shift={(9.5,0)}]
\tikzstyle{every node} = [draw,shape=circle]
\node (1) at (1,2) {$abdef^T$};
\node (3) at (0,0) {$c^H$};
\foreach \from/\to in {1/3}
  \draw [line width=2pt, ->] (\from) -- (\to);
\end{scope}

\node at (12.5,2) {\Huge$\Ra$};

\begin{scope}[shift={(12.5,0)}]
\tikzstyle{every node} = [draw,shape=circle]
\node (1) at (2,2) {$abcdef$};
\end{scope}
\end{tikzpicture}
\caption{{\tt HeadsTailsParallel} on a sample graph with original vertices $a$ through $f$.}\label{fig:HTParallel}
\end{figure}
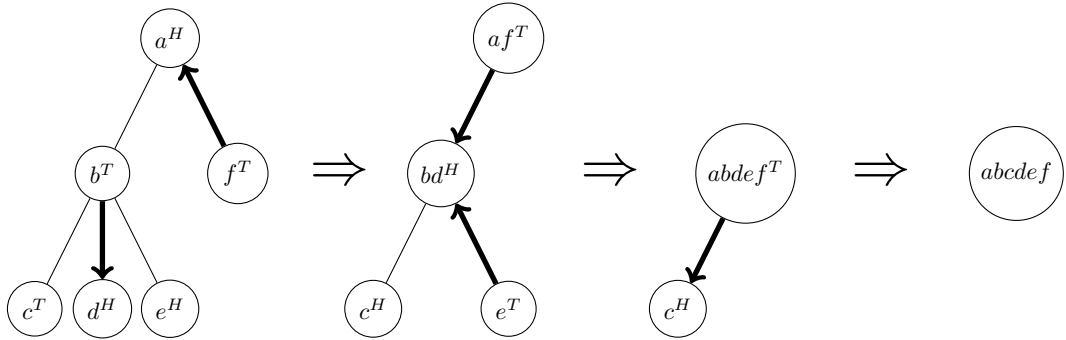

\begin{proof}
Suppose that the algorithm runs for $K:=\Th(\logn)$ rounds, disregarding the stopping condition. As long as there is more than one vertex remaining,  every vertex has at least one neighbor. The probability that a given vertex flips Tails and a specific neighbor flips Heads is $\f14$, so every vertex merges with a neighbor with probability at least $\f14$. Therefore, if $n_i$ is the number of vertices left on iteration $i$ ($n_i\ge1$), then $\E[n_{i+1}|n_i=k]-1\le\f34(k-1)$, and applying induction gives $\E[n_K]-1 \le (\f34)^K(n-1)=\f1{\poly(n)}$. By Markov's inequality, the probability that $n_K\ge2$ is $\f1{\poly(n)}$, so w.h.p. the stopping condition activates within $K$ rounds.

It is easy to see that (1) every cluster output by the algorithm is connected, and (2) $P_0$ and $P_K$ are the singleton partition and the single-cluster partition.

For a cluster $C_{i,j}$, consider the set of vertices $v_{k_1},v_{k_2},\lds$ at the beginning of iteration $i$ which contract to vertex $j$ on that iteration. There is at most one Heads vertex, and every Tails is adjacent to the Heads, so together, these vertices form a star component of diameter at most 2. Therefore, the clusters $C_{i-1,k_1},C_{i-1,k_2},\lds$ partition $C_{i,j}$ and, when contracted in $C_{i,j}$, form a graph of diameter $O(1)$.
\end{proof}

Now we present the Heads/Tails algorithm in the distributed setting. Observe that, since we can no longer contract connected components into single vertices, we instead construct shortcuts to allow for efficient communication between components.

\begin{mdframed}[roundcorner=4pt, backgroundcolor=white]
\textbf{Algorithm {\tt{HeadsTailsDistributed}}}
\be
\im Compute $\tilde{O}(Q_G)$-quality shortcuts using each cluster as a part.
\im For each cluster, choose a ``leader'' as follows: every vertex within the cluster computes the minimum vertex ID over that cluster using shortcuts. The vertex with this ID is the leader.
\im For each cluster, the leader flips either Heads or Tails and broadcasts this bit to the entire cluster.
\im If an edge $(h,t)$ in between two clusters has $h$ in a Heads cluster and $t$ in a Tails, then send the message $(h,t)$ to vertex $t$.
\im For every Tails cluster, every vertex computes the minimum (lexicographic) message received, if any, over that cluster.
\im If a vertex $t$ in a Tails cluster $C$ receives message $(h,t)$, then broadcast to $C$ that $C$ has merged with the cluster containing $h$ (i.e. they are now the same part).
\im Every cluster computes its size and stops if it equals the size of $T$. Otherwise, repeat from the start.
\ee
\end{mdframed}

\begin{lemma}
{\tt{HeadsTailsParallel}} terminates in $\tilde O(Q_G)$ rounds w.h.p., and outputs a low-diameter hierarchical clustering.\end{lemma}

\begin{proof}
We compare {\tt{HeadsTailsDistributed}} to {\tt{HeadsTailsParallel}}. Steps (2) and (3) emulate a contracted vertex flipping a bit. Steps (4) through (6) correspond to choosing an arbitrary Heads neighbor for each Tails cluster. Finally, Step (7) checks to see if only one cluster remains. Note that the sizes can be computed efficiently by aggregating $\Sum_{v\in C}1$ within each cluster. It is also clear that in each iteration, either no cluster terminates, or all of them do.  As for running time, every iteration has $O(1)$ aggregates and $O(1)$ broadcasts, each running in $\tilde{O}(Q_G)$ time, so over all $O(\logn)$ iterations w.h.p., the clustering algorithm takes $\tilde{O}(Q_G)$ time.
\end{proof}

Observe that the algorithm is careful to ensure that two vertices do not broadcast within the same cluster simultaneously. In the future, we will skip the more pedantic steps in aggregating and broadcasting.

\subsection{Aggregate Functions}
As before, let $G=(V,E)$ be a distributed network and let $T\subseteq G$ be a tree. Consider rooting the tree at a predetermined root $r$, and suppose every vertex $v$ in the tree has a value $x_v$. For a commutative function $\oplus$, we want to compute, for every $v\in V[T]$, the value $\displaystyle\bigoplus\limits_{u\in S_v} x_u$ for two types of vertex sets $S_v\subseteq V[T]$: the set of vertices on the path from $v$ to the root, and the set of vertices in the subtree rooted at $v$. Using shortcuts of quality $Q_G$, we show how to compute these two aggregates in time $\tilde{O}(Q_G)$. At a high level, our algorithm mimics the divide-and-conquer algorithms in parallel computing, with the clusters representing the divided inputs. Figure \ref{fig:AggregateSubtree} illustrates {\tt AggregateSubtree} on a sample graph.

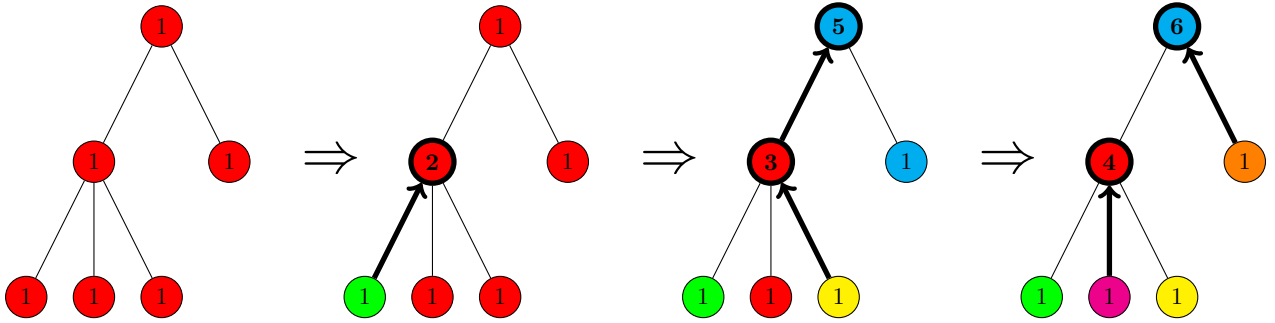
\begin{figure}
\centering
\begin{tikzpicture}[scale=.9, transform shape]

\begin{scope}
\tikzstyle{every node} = [draw,shape=circle]
\node[fill=red] (1) at (2,4) {1};
\node[fill=red] (2) at (1,2) {1};
\node[fill=red] (3) at (0,0) {1};
\node[fill=red] (4) at (1,0) {1};
\node[fill=red] (5) at (2,0) {1};
\node[fill=red] (6) at (3,2) {1};
\foreach \from/\to in {1/2,2/3,2/5,2/4,1/6}
  \draw [] (\from) -- (\to);
\end{scope}

\node at (4.5,2) {\Huge$\Ra$}; 

\begin{scope}[shift={(5,0)}]
\tikzstyle{every node} = [draw,shape=circle]
\node[fill=red] (1) at (2,4) {1};
\node[fill=red,line width=2pt] (2) at (1,2) {\textbf 2};
\node[fill=green] (3) at (0,0) {1};
\node[fill=red] (4) at (1,0) {1};
\node[fill=red] (5) at (2,0) {1};
\node[fill=red] (6) at (3,2) {1};
\foreach \from/\to in {1/2,2/3,2/5,2/4,1/6}
  \draw [] (\from) -- (\to);
\foreach \from/\to in {3/2}
  \draw [line width=2pt,->] (\from) -- (\to);
\end{scope}

\node at (9.5,2) {\Huge$\Ra$};

\begin{scope}[shift={(10,0)}]
\tikzstyle{every node} = [draw,shape=circle]
\node[fill=cyan,line width=2pt] (1) at (2,4) {\textbf 5};
\node[fill=red,line width=2pt] (2) at (1,2) {\textbf 3};
\node[fill=green] (3) at (0,0) {1};
\node[fill=red] (4) at (1,0) {1};
\node[fill=yellow] (5) at (2,0) {1};
\node[fill=cyan] (6) at (3,2) {1};
\foreach \from/\to in {1/2,2/3,2/5,2/4,1/6}
  \draw [] (\from) -- (\to);
\foreach \from/\to in {2/1,5/2}
  \draw [line width=2pt,->] (\from) -- (\to);
\end{scope}

\node at (14.5,2) {\Huge$\Ra$};

\begin{scope}[shift={(15,0)}]
\tikzstyle{every node} = [draw,shape=circle]
\node[fill=cyan,line width=2pt] (1) at (2,4) {\textbf 6};
\node[fill=red,line width=2pt] (2) at (1,2) {\textbf 4};
\node[fill=green] (3) at (0,0) {1};
\node[fill=magenta] (4) at (1,0) {1};
\node[fill=yellow] (5) at (2,0) {1};
\node[fill=orange] (6) at (3,2) {1};
\foreach \from/\to in {1/2,2/3,2/5,2/4,1/6}
  \draw [] (\from) -- (\to);
\foreach \from/\to in {4/2,6/1}
  \draw [line width=2pt,->] (\from) -- (\to);
\end{scope}

\end{tikzpicture}

\caption{{\tt AggregateSubtree} on a sample graph with initial $x_v=1$, using the partitions from Figure \ref{fig:HTParallel}. The node labels represent the current value of $x_v$ in the algorithm. The color classes represent the current parts. The arrows indicate the updates in step 4 of the algorithm.}\label{fig:AggregateSubtree}
\end{figure}

\begin{mdframed}[roundcorner=4pt, backgroundcolor=white]
\textbf{Algorithm {\tt{AggregateTree}}}
\be
\im Let $\mc P_0,\lds,\mc P_\el$ be the partitions output by the Heads/Tails algorithm, where $\mc P_0$ is the singleton vertices and $\mc P_\el$ is one entire component. Set $r$ to be the root of $\mc P_\el$.
\im Initially, each $v\in V[T]$ has value $x_v$. \ee

        \im \textbf{Sub-algorithm {\tt AggregateSubtree}}:
\bi \im For $t$ from $\el-1$ down to $0$:
\bi \im For each cluster $C\in \mc P_{t+1}$:
 \be
 \im Let the cluster in $\mc P_{t}$ that are contained in $C$ be $\{
C_1,\lds,C_k\}$.
 \im For each cluster $C_i$, compute the aggregate $F(C_i):=\displaystyle\bigoplus\limits_{u\in C_i} x_u$.
 \im Let the root of $C$ be $r(C)$, and assume that $C_1$ contains $r(C)$. Set $r(C_1):=r(C)$. Within cluster $C$, view the $C_i$ as contracted vertices in a tree, and run BFS from $C_1$. For each $C_i$, $i>1$, set $r(C_i)$ to be the vertex in $C_i$ adjacent to the parent of $C_i$ in the BFS tree.
 \im For every edge $(u,v)$ in between two parts where the BFS travels from $u$ to $v$, set $x_v := x_v \oplus F(C_i)$.
 \ee
\ei \ei
        \im \textbf{Sub-algorithm {\tt AggregatePathToRoot}}:
\bi \im For $t$ from $0$ up to $\el-1$:
\bi \im for each part $P\in C_{t+1}$:
 \be
 \im Let the root of $C$ be $r(C)$, and assume that $C_1$ contains $r(C)$. Set $r(C_1):=r(C)$. Within cluster $C$, view the $C_i$ as contracted vertices in a tree, and run BFS from $C_1$. For each $C_i$, $i>1$, set $r(C_i)$ to be the vertex in $C_i$ adjacent to the parent of $C_i$ in the BFS tree.
 \im For every edge $(u,v)$ in between two parts where the BFS travels from $u$ to $v$, send the value $x_u$ to $v$.
 \im For each part except the root, the one vertex that receives some value $x$ broadcasts it to its entire part. Then, every vertex in the part applies $x_v := x_v \oplus x$.
 \ee
\ei \ei
\end{mdframed}

It is easy to see, by induction on $t$, that the two algorithms are correct. The running time is clearly $\tilde{O}(Q_G)$.

\section{Low-diameter Decompositions}\label{sec:ldd}

In this section, we present the exponential starting time algorithm by Miller et al. \cite{miller2013parallel}, with a few modifications:
 \be
 \im We work on weighted graphs with weights in the range $[1,\poly(n)]$.
 \im We work under a distributed setting, where the graph is the network.
 \im We replace exponential random variables with their discrete cousin, the geometric. The latter does not have to deal with rounding real numbers, and yet maintains the special \textit{memoryless} property crucial to the algorithm.
 \ee
 Instead of referring to \cite{miller2013parallel}, we prove all the properties of the LDD that we need, since the proofs are simple and allow this paper to become self-contained.

\begin{mdframed}[roundcorner=4pt, backgroundcolor=white]

\textbf{Algorithm {\tt{LowDiameterDecomposition}}}

Input:
 \bi
 \im Weighted graph $G$ with edge weights in the range $[1,\poly(n)]$.
 \im Parameter $\bt$, freely chosen.
 \ei

Output: A partition of $V$ into connected components such that:
 \bi
 \im W.h.p., each component has (strong) diameter $O(\f1\bt\logn)$.
 \im For vertices $u,v\in V$ of (weighted) distance $d$, the probability that $u$ and $v$ belong to the same cluster is $e^{-O(d\bt)}$.
 \ei

Algorithm:
 \be
 \im Every vertex $u$ picks $\de_u$ independently from a geometric distribution with parameter $\bt$.
 \im For each vertex $u$, set its starting time to $t_u:=\f c\bt\logn-\de_u$, where $c$ is a predetermined, large-enough constant. (W.h.p., no $\de_u$ will exceed $c\logn$.)
 \im Simulate a continuous-time, parallel BFS with vertex $u$ starting at time $t_u$ for $\f c\bt\logn$ rounds.
 \im Whenever the BFS reaches vertex $u$, assign $u$ to the root vertex of this BFS path. The resulting components are the sets of vertices assigned to a common root.
 \ee
\end{mdframed}

\begin{figure}
\centering
\begin{tikzpicture}[scale=.9, transform shape]

\begin{scope}
\tikzstyle{every node} = [draw,shape=circle]
\node[] (1) at (0,2) {1};
\node[] (2) at (2,0) {1};
\node[] (3) at (4,1) {4};
\node[] (4) at (4,3) {2};
\node[fill=red,line width=2pt] (5) at (1,4) {0};
\node[] (6) at (3,5) {2};
\tikzstyle{every node} = []
\foreach \from/\to/\len in {1/2/0.1,2/3/0.9,4/5/1,5/6/4}
  \draw [] (\from) -- (\to) node[pos=.5,above]{\len};
\foreach \from/\to/\len in {3/4/1,2/5/1.1}
  \draw [] (\from) -- (\to) node[pos=.5,right]{\len};
\foreach \from/\to/\len in {5/1/0.9}
  \draw [] (\from) -- (\to) node[pos=.5,left]{\len};
\foreach \from/\to in {}
  \draw [line width=2pt, ->] (\from) -- (\to);
\node at (2,-1) {\huge $t=0$};
\end{scope}

\begin{scope}[shift={(6,0)}]
\tikzstyle{every node} = [draw,shape=circle]
\node[fill=red,line width=2pt] (1) at (0,2) {1};
\node[fill=green,line width=2pt] (2) at (2,0) {1};
\node[] (3) at (4,1) {4};
\node[fill=red,line width=2pt] (4) at (4,3) {2};
\node[fill=red] (5) at (1,4) {0};
\node[] (6) at (3,5) {2};
\tikzstyle{every node} = []
\foreach \from/\to/\len in {1/2/0.1,2/3/0.9,4/5/1,5/6/4}
  \draw [] (\from) -- (\to) node[pos=.5,above]{\len};
\foreach \from/\to/\len in {3/4/1,2/5/1.1}
  \draw [] (\from) -- (\to) node[pos=.5,right]{\len};
\foreach \from/\to/\len in {5/1/0.9}
  \draw [] (\from) -- (\to) node[pos=.5,left]{\len};
\foreach \from/\to in {5/1,5/4,1/2}
  \draw [line width=2pt, ->] (\from) -- (\to);
\node at (2,-1) {\huge $t=1$};
\end{scope}

\begin{scope}[shift={(12,0)}]
\tikzstyle{every node} = [draw,shape=circle]
\node[fill=red] (1) at (0,2) {1};
\node[fill=green] (2) at (2,0) {1};
\node[fill=green,line width=2pt] (3) at (4,1) {4};
\node[fill=red] (4) at (4,3) {2};
\node[fill=red] (5) at (1,4) {0};
\node[fill=cyan,line width=2pt] (6) at (3,5) {2};
\tikzstyle{every node} = []
\foreach \from/\to/\len in {1/2/0.1,2/3/0.9,4/5/1,5/6/4}
  \draw [] (\from) -- (\to) node[pos=.5,above]{\len};
\foreach \from/\to/\len in {3/4/1,2/5/1.1}
  \draw [] (\from) -- (\to) node[pos=.5,right]{\len};
\foreach \from/\to/\len in {5/1/0.9}
  \draw [] (\from) -- (\to) node[pos=.5,left]{\len};
\foreach \from/\to in {2/3}
  \draw [line width=2pt, ->] (\from) -- (\to);
\node at (2,-1) {\huge $t=2$};
\end{scope}

\end{tikzpicture}

\caption{{\tt LowDiameterDecomposition} on a sample graph. The node labels are the starting times $t_u$. The color classes represent the current components at time $t$. Note that there is a tie on the green node at $t=1$.}\label{fig:LDD}
\end{figure}
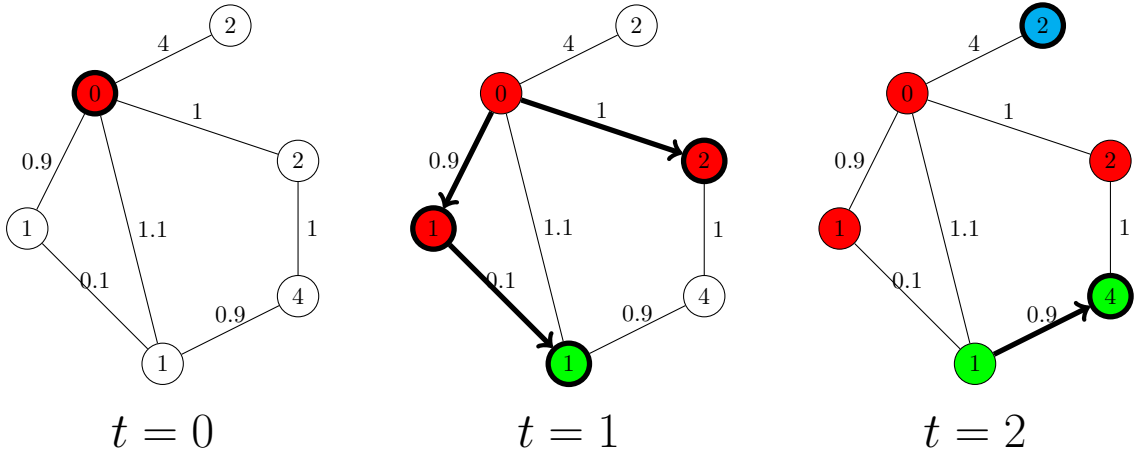

The algorithm is illustrated on a sample graph in Figure \ref{fig:LDD}. First, we need to describe how to simulate a continuous-time BFS on a distributed network with edge weights at least $1$. Throughout the simulation, we maintain the invariant that a node that is reached at time $t$ will receive the message ``$t$'' at time $\lf t\rf$. When this happens, the node sends along each incident edge $e$ the message ``$t+w_e$'' at round $\lf t+w_e\rf$. Since $w_e\ge1$, $\lf t+w_e\rf > \lf t\rf$, so the round in which to send the message is always in the future. Note that a node may receive multiple messages ``$t_i$'' on a single round, in which case it is occupied at the earliest time $t_i$.

\LDDDiam*

\begin{proof}
The probability that some $t_u=\f c\bt\logn-\de_u$ is negative is the probability that  $\text{Geo}(\bt)>\f c\bt\logn$, which is $(1-\bt)^{(c/\bt)\logn}\le n^{-c}$. Otherwise, since every vertex is assigned a nonnegative time bounded by $\f c\bt\logn$, and the BFS trees out of each root have diameter at most $\f c\bt\logn$, it follows that each component also has diameter at most $\f c\bt\logn$.
\end{proof}

\LDDProb*

\begin{proof}
For each root $r$, consider the random variable $X_{r,u}:=t_r+d(r,u)$, which is the time when the BFS at $r$ would reach $u$ if no other BFS paths interfere. We claim that if the difference between the smallest and second-smallest $X_{r,u}$ (call them $X_{r,u}^{(1)}$ and $X_{r,u}^{(2)})$ over all roots $r$ are more than $2d$ apart, then $u$ and $v$ belong in the same component. Suppose for contradiction that $u$ and $v$ belong to different components. Then, the times when the BFS reaches $u$ and $v$ must be within $d$ of each other, since otherwise, the first BFS tree to each one of $u$ and $v$ would also reach the other vertex first. The BFS tree that reaches $v$ first takes at most another $d$ time to reach $u$, which means that $X_{r,u}^{(2)}-X_{r,u}^{(1)}\le2d$, contradiction.

To bound the probability that $X_{r,u}^{(2)}-X_{r,u}^{(1)}>2d$, consider the variable $X_{r,u}^{(1)}$ conditioned on the value of $X_{r,u}^{(2)}$. If $r$ is the root that produces $X_{r,u}^{(1)}$, then we are looking for the value of $X_{r,u}^{(2)} - (t_r+d(r,u))=\de_r-(\f c\bt\logn-X_{r,u}^{(2)}+d(r,u))$ conditioned on the value being nonnegative. This is the same value as $\de_r-C$ conditioned on the event $\de_r\ge \lc C\rc$, for the (possibly negative) $C:=c\bt\logn-X_{r,u}^{(2)}+d(r,u)$. We can focus on the probability that $\de_r-\lc C\rc \ge \lc 2d\rc$, since this implies that $\de_r-C>2d$. By the memoryless property of geometric random variables, the value of $\de_r-\lc C\rc$ conditioned on $\de_r\ge \lc C\rc$ follows the distribution $\text{Geo}(\bt)+\max\{0,-\lc C\rc\}$, and the desired probability is at least $(1-\bt)^{\lc2d\rc}\ge e^{-\bt\cd O(d)}$.
\end{proof}

\end{document}